\documentclass[11pt]{article}

\usepackage[margin=1in]{geometry}
\usepackage{graphicx}
\usepackage[pdftex,colorlinks=true,linkcolor=blue,citecolor=blue,urlcolor=black]{hyperref}
\usepackage{amsmath, amsthm, amssymb}
\usepackage{subfigure}
\usepackage{comment}
\usepackage{url}
\usepackage{pdflscape}
\usepackage[ruled,lined,linesnumbered]{algorithm2e}
\usepackage{tikz}
\usepackage{multicol}

\newcommand{\R}{\mathbb{R}}

\newcommand{\C}{\mathbb{C}}

\newcommand{\E}{\mathbb{E}}

\newcommand{\ket}[1]{| #1 \rangle}

\newcommand{\ip}[2]{\langle #1|#2 \rangle}

\newcommand{\proj}[1]{| #1 \rangle \langle #1 |}

\DeclareMathOperator{\poly}{poly}

\DeclareMathOperator{\spann}{span}

\newcommand{\be}{\begin{equation}}
\newcommand{\ee}{\end{equation}}
\newcommand{\bea}{\begin{eqnarray}}
\newcommand{\eea}{\end{eqnarray}}
\newcommand{\bes}{\begin{equation*}}
\newcommand{\ees}{\end{equation*}}
\newcommand{\beas}{\begin{eqnarray*}}
\newcommand{\eeas}{\end{eqnarray*}}

\makeatletter
\newtheorem*{rep@theorem}{\rep@title}
\newcommand{\newreptheorem}[2]{%
\newenvironment{rep#1}[1]{%
 \def\rep@title{#2 \ref{##1} (restated)}%
 \begin{rep@theorem}}%
 {\end{rep@theorem}}}
\makeatother

\newtheorem{thm}{Theorem}
\newtheorem*{thm*}{Theorem}

\newtheorem{lem}[thm]{Lemma}

\newtheorem*{lem*}{Lemma}
\newtheorem{prop}[thm]{Proposition}

\newtheorem{fact}[thm]{Fact}

\newreptheorem{thm}{Theorem}
\newreptheorem{lem}{Lemma}

\newcommand{\boxalgm}[3]{
\renewcommand{\figurename}{Algorithm}
\begin{figure}[tb]
\begin{center}
\noindent \framebox{
\begin{minipage}{.95\textwidth}
#3
\end{minipage}
}
\caption{#2}
\label{#1}
\end{center}
\end{figure}
\renewcommand{\figurename}{Figure}
}

\begin{document}

\title{Quantum walk speedup of backtracking algorithms}
\author{Ashley Montanaro\thanks{School of Mathematics, University of Bristol, UK; {\tt ashley.montanaro@bristol.ac.uk}.}}
\maketitle

\begin{abstract}
We describe a general method to obtain quantum speedups of classical algorithms which are based on the technique of backtracking, a standard approach for solving constraint satisfaction problems (CSPs). Backtracking algorithms explore a tree whose vertices are partial solutions to a CSP in an attempt to find a complete solution. Assume there is a classical backtracking algorithm which finds a solution to a CSP on $n$ variables, or outputs that none exists, and whose corresponding tree contains $T$ vertices, each vertex corresponding to a test of a partial solution. Then we show that there is a bounded-error quantum algorithm which completes the same task using $O(\sqrt{T} n^{3/2} \log n)$ tests. In particular, this quantum algorithm can be used to speed up the DPLL algorithm, which is the basis of many of the most efficient SAT solvers used in practice. The quantum algorithm is based on the use of a quantum walk algorithm of Belovs to search in the backtracking tree. We also discuss how, for certain distributions on the inputs, the algorithm can lead to an exponential reduction in expected runtime.
\end{abstract}


\section{Introduction}

Grover's quantum search algorithm~\cite{grover97} is one of the great success stories of quantum computation. One important domain to which the algorithm can be applied is the solution of constraint satisfaction problems (CSPs). Consider a constraint satisfaction problem (CSP) expressed as a predicate $P:[d]^n \rightarrow \{\text{true},\text{false}\}$, where $[d] = \{0,\dots,d-1\}$. We would like to find an assignment $x$ to the $n$ variables such that $P(x)$ is true, or output ``not found'' if no such $x$ exists. This framework encompasses many important problems such as boolean satisfiability and graph colouring. Grover's algorithm solves such a CSP using $O(\sqrt{d^n})$ evaluations of $P$, whereas with no further information about $P$, finding an $x$ such that $P(x)$ is true requires $\Omega(d^n)$ evaluations classically in the worst case.
However, when we are faced with an instance of a CSP in practice, we usually have some additional information about its structure. For example, $P$ may be defined as the conjunction of smaller constraints of a particular type, as in the case of graph colouring. This information often allows classical algorithms to solve the CSP significantly more efficiently than the above bound would suggest, throwing some doubt on whether straightforward use of Grover's algorithm will really be used to solve CSPs in practice.

One of the most important and most general classical tools to take advantage of problem structure, both in theory and in practice, is backtracking~\cite{vanbeek06}. This technique can be used when we have the ability to recognise whether partial solutions to a problem can be extended to full solutions. We assume that the predicate $P$ allows us to pass it a partial assignment $x$ of the form $x:S \rightarrow [d]$, where $S \subseteq \{1,\dots,n\}$, which specifies the values assigned to the variables in the set $S$. We can equivalently think of $x$ as an element of $\mathcal{D} := ([d] \cup \{\ast\})^n$, where the $\ast$'s represent the positions which are as yet unassigned values. We say that $x$ is {\em complete} if it contains no $\ast$'s. Then $P$ returns ``true'' if $x$ is a solution to $P$, ``false'' if it is clear that $x$ cannot be extended to a solution to $P$, and ``indeterminate'' otherwise. We say that a partial assignment $x$ is valid if $P(x)$ is true or indeterminate, and invalid if $P(x)$ is false.

Algorithm \ref{alg:backtrack} above describes a generic way to use this information classically. The algorithm assumes access to $P$ and a heuristic $h(x)$ which determines how to extend a given partial assignment $x$. We think of $P$ and $h$ as black boxes (``oracles''). The basic idea is to fail early: if we know that a partial assignment cannot be extended to a solution, we should give up on it and try a different one. We can think of the algorithm as exploring a tree, whose internal vertices are partial solutions to $P$, and whose leaves are solutions to $P$ or certificates that the partial solution cannot be extended to a complete solution. This tree is of size at most $O(d^n)$, but for some problem instances could be substantially smaller. 

A canonical example of a powerful backtracking algorithm which fits into the framework of Algorithm \ref{alg:backtrack} is the DPLL (Davis-Putnam-Logemann-Loveland) algorithm~\cite{davis60,davis62} for $k$-SAT. This algorithm forms the basis of many of the most successful SAT solvers used in practice~\cite{een03,lynce03,gomes08}. For many practically relevant problem instances, the algorithm runs more quickly than worst-case upper bounds would suggest. Another appealing aspect of this algorithm is that, unlike ``local search'' methods based on random walks or similar ideas, it can sometimes produce efficient proofs of unsatisfiability, corresponding to small backtracking trees.

\boxalgm{alg:backtrack}{General classical backtracking algorithm}{
Assume that we are given access to a predicate $P: \mathcal{D} \rightarrow \{ \text{true}, \text{false}, \text{indeterminate} \}$, and a heuristic $h: \mathcal{D} \rightarrow \{1,\dots,n\}$ which returns the next index to branch on from a given partial assignment.
\\[3pt]
Return {\tt bt}$(\ast^n)$, where {\tt bt} is the following recursive procedure:
\\[3pt]
{\tt bt}$(x)$:
\begin{enumerate}
\item If $P(x)$ is true, output $x$ and return.
\item If $P(x)$ is false, or $x$ is a complete assignment, return.
\item Set $j = h(x)$.
\item For each $w \in [d]$:
\begin{enumerate}
\item Set $y$ to $x$ with the $j$'th entry replaced with $w$.
\item Call {\tt bt}$(y)$.
\end{enumerate}
\end{enumerate}
}

Algorithm \ref{alg:backtrack} outputs all solutions $x$ such that $P(x)$ is true. While in practice the algorithm might be modified to terminate when the first solution is found, here we will assume throughout that the entire tree is explored. We assume that $P$ and $h$ can both be evaluated in time $\poly(n)$, so the most important contribution to the complexity of Algorithm \ref{alg:backtrack} is usually the number of vertices in the tree, which can often be exponential in $n$. To simplify the complexity bounds, we also assume throughout that $d=O(1)$; this is effectively without loss of generality as any predicate with local domain size $d$ can be replaced with one which uses $O(\log d)$ bits to encode each variable.




\subsection{Results}

We show here that there is a quantum equivalent of Algorithm \ref{alg:backtrack} which can be substantially faster:

\begin{thm}
\label{thm:detectbound}
Let $T$ be an upper bound on the number of vertices in the tree explored by Algorithm \ref{alg:backtrack}. Then for any $0 < \delta < 1$ there is a quantum algorithm which, given $T$, evaluates $P$ and $h$ $O(\sqrt{Tn} \log (1/\delta))$ times each, outputs true if there exists $x$ such that $P(x)$ is true, and outputs false otherwise. The algorithm uses $\poly(n)$ space, $O(1)$ auxiliary operations per use of $P$ and $h$, and fails with probability at most $\delta$.
\end{thm}

We usually think of $T$ as being exponential in $n$; in this regime this complexity is a near-quadratic speedup over the classical algorithm. The algorithm can be modified to find a solution, rather than just detect the existence of one, with a small runtime penalty:

\begin{thm}
\label{thm:findbound}
Let $T$ be the number of vertices in the tree explored by Algorithm \ref{alg:backtrack}. Then for any $0 < \delta < 1$ there is a quantum algorithm which evaluates $P$ and $h$ $O(\sqrt{T}n^{3/2} \log n \log (1/\delta))$ times each, and outputs $x$ such that $P(x)$ is true, or ``not found'' if no such $x$ exists. If we are promised that there exists a unique $x_0$ such that $P(x_0)$ is true, there is a quantum algorithm which outputs $x_0$ using $P$ and $h$ $O(\sqrt{Tn} \log^3 n \log(1/\delta))$ times each. In both cases the algorithm uses $\poly(n)$ space, $O(1)$ auxiliary operations per use of $P$ and $h$, and fails with probability at most $\delta$.
\end{thm}

We stress that these results can be applied to any backtracking algorithm which fits into the framework of Algorithm~\ref{alg:backtrack}, whatever the predicate $P$ or the choice of the heuristic $h$. In particular, they can be applied to the DPLL algorithm with the commonly used ``unit clause'' heuristic. Theorems \ref{thm:detectbound} and \ref{thm:findbound} can also be applied to backtracking algorithms which make use of randomness in the heuristic $h$, by interpreting these algorithms as first fixing a random seed, then using this seed as input to a deterministic heuristic $h$. Observe that the runtime bound of Theorem \ref{thm:findbound} is instance-dependent and, to use it, we do not need to know an upper bound on the runtime $T$ of the underlying classical backtracking algorithm. For instances on which the classical algorithm runs quickly, the quantum algorithm also runs quickly.

These algorithms can be leveraged to obtain an {\em exponential} separation between average quantum and classical runtimes. The speedup for any given instance is approximately quadratic. However, given the right distribution on the input instances, this can be amplified to an exponential separation. This is discussed further in Section \ref{sec:averagecase}.


\subsection{Techniques}

The algorithms which achieve the bounds of Theorems \ref{thm:detectbound} and \ref{thm:findbound} are based on the use of a discrete-time quantum walk to find a marked vertex within the tree produced by the classical backtracking algorithm, corresponding to a partial solution $x$ such that $P(x)$ is true. Quantum walks have become a basic tool in quantum algorithm design~\cite{childs03,ambainis04,szegedy04,magniez11}. In particular, they have been applied in several contexts to solve search problems on graphs~\cite{shenvi03,szegedy04,magniez11,krovi15}, sometimes achieving up to a quadratic speedup over classical algorithms. However, in prior work it is usually assumed that the input graph is known in advance, and moreover that the initial state of the quantum walk is the stationary distribution of the corresponding random walk. Aaronson and Ambainis~\cite{aaronson05} described a different approach to spatial search on graphs; this does not use a quantum walk, but also assumes the input graph is known in advance.

Here we would like to use quantum walks in a context where the input graph is defined implicitly by the backtracking algorithm and hence is not known in advance, and where the walk starts at the root of the tree. One of the few cases where such walks have been studied is beautiful work of Belovs~\cite{belovs13,belovs13a}. The main result of that work relates the complexity of detecting a marked vertex by quantum walk on a graph to the {\em effective resistance} of the graph. Informally, this quantity is determined by thinking of the graph as an electrical circuit and calculating the resistance between the initial vertex and the set of marked vertices. Belovs' result can be seen as a quantum variant of previous classical work characterising properties of random walks on graphs (such as the commute time and cover time) in terms of effective resistance~\cite{chandra97}.

The main quantum subroutine used here is just the special case of Belovs' result where the underlying graph is a tree, for which we include a slightly more concise correctness proof. We are also able to extend Belovs' work to give an algorithm for finding a marked vertex in a tree, rather than just detecting one. This can easily be achieved using binary search; in the case where there is promised to be a unique marked element, we give a more efficient algorithm based on analysing eigenvectors of the quantum walk.

Once we have the quantum search algorithm, all that remains is to check the claim that the $P$ and $h$ functions can indeed be used to implement the required quantum walk operations, namely mixing across the neighbours of a vertex in the tree, dependent on whether the vertex is marked. To do this one has to be careful to ensure that the quantum walk steps are implemented efficiently.


\subsection{Other prior work}

Backtracking is a fundamental technique in computer science and has been studied since at least the 1960s. The classical literature on this topic is too vast to summarise here; see~\cite{knuth75,freuder06,vanbeek06} for introductions to the topic and historical overviews. Cerf, Grover and Williams attempted to find a direct quantum speedup of backtracking algorithms~\cite{cerf00}. The algorithm of~\cite{cerf00} is based on a nested version of Grover search. The complete tree of partial assignments is expanded to a certain depth, then quantum search is performed within the subspace of partial assignments which have not yet been ruled out. The complexity of the algorithm depends on the number of valid partial assignments at this depth. It is argued in~\cite{cerf00} that, for some reasonable distributions on random CSPs, the average complexity of the quantum algorithm (over the distribution on instances) will be smaller than would be obtained from Grover search. By contrast, the bounds of Theorems \ref{thm:detectbound} and \ref{thm:findbound} hold in the worst case and are applicable to arbitrary backtracking algorithms: if a faster backtracking algorithm is found, we immediately obtain a faster quantum algorithm.

The algorithm used here can be seen as an extreme version of the nested search strategy of~\cite{cerf00}. The diffusion operation used in the quantum walk can be viewed as applying Grover search within a subspace spanned by a vertex in the tree and its children. The algorithm repeatedly performs these searches across many vertices and levels simultaneously. On the other hand, the algorithm of~\cite{cerf00} can be seen as accelerating a restricted classical backtracking algorithm which uses a predicate $P$ which is only capable of detecting whether partial assignments at a particular level are false.

Similarly to the present work, Farhi and Gutmann~\cite{farhi98} have studied the use of quantum walks to speed up classical backtracking algorithms by searching within the backtracking tree. These authors showed that there are some trees for which continuous-time quantum walks can be used to find a marked vertex exponentially faster than a classical random walk. The special structure of these trees leads to interference effects which enable the quantum walk to penetrate the tree more quickly than the random walk. However, for the examples presented in~\cite{farhi98} where there is an exponential speedup of this form, the structure of the tree enables an alternative classical algorithm to also find a marked vertex efficiently. Here, we seek to accelerate classical search in arbitrary trees, with no prior assumptions about the structure of the tree.

A related, but different, approach towards quantum speedup of recursive classical algorithms was proposed by F\"urer~\cite{furer08}. Imagine we have a constraint satisfaction problem for which we can put a non-trivial upper bound $L$ on the number of leaves in the computation tree of a recursive classical algorithm for solving the problem. The idea of~\cite{furer08} was to apply Grover search over the leaves of the computation tree to find a solution in time $O(\sqrt{L} \poly(n))$. This approach relies on knowing, in advance, an efficiently computable mapping associating each integer between 1 and $L$ with a leaf. For many more complicated recursive algorithms we may not know such a mapping. Indeed, there is some evidence that it may not be possible to compute such a mapping for general backtracking algorithms in polynomial time~\cite{stockmeyer85}. The quantum algorithm presented here, on the other hand, can be applied to any classical backtracking algorithm, even if we do not know a bound on $L$ in advance.

A somewhat similar idea to F\"urer's was previously used by Angelsmark, Dahll\"of and Jonsson~\cite{angelsmark02} to obtain quantum speedups for CSPs. These authors observed that, for certain CSPs, one can construct a set of $d^{cn}$ easily checked certificates, for some $c < 1$, such that the existence of a solution to the CSP is certified by at least one certificate. Then Grover search can be used to find a certificate, if one exists, in time $O(d^{cn/2} \poly(n))$.

An alternative, and simpler, approach to find quantum speedups of classical algorithms for CSPs is the use of amplitude amplification~\cite{brassard02}. This can be applied to any classical algorithm which can be expressed as repeatedly running a randomised subroutine which runs in time $\poly(n)$ and finds a solution with probability $p$. The corresponding quantum algorithm has a runtime of $O((1/\sqrt{p}) \poly(n))$, a near-quadratic improvement on the classical $O((1/p)\poly(n))$ if $p$ is small. For example, it was observed by Ambainis~\cite{ambainis04} that Sch\"oning's efficient randomised algorithm for $k$-SAT~\cite{schoning99} can be accelerated in this way; Dantsin, Kreinovich and Wolpert~\cite{dantsin05} gave several other examples. Deterministic backtracking algorithms are, of course, not amenable to this approach.

Finally, a completely different technique for solving CSPs is the quantum adiabatic algorithm~\cite{farhi00}. Although there is some numerical evidence that this algorithm may outperform classical algorithms for CSPs~\cite{farhi01}, the adiabatic algorithm's runtime is hard to analyse for large input sizes and there is as yet no analytical proof of its superiority over classical algorithms.

Quantum walks on trees have been used previously in a quite different context, to obtain a near-quadratic speedup for evaluation of AND-OR formulae~\cite{ambainis10b}. In that algorithm the structure of the formula (which is known in advance) defines the tree on which the walk takes place. It is interesting to note that the quantum walk used in~\cite{ambainis10b} is similar to the quantum walk used here, but has apparently quite different properties. Another case in which the concept of effective resistance was used in quantum computing is work by Wang, which gave an efficient quantum algorithm for approximating effective resistances~\cite{wang13}. This uses some similar ideas to the present work but does not seem directly applicable.


\subsection{Organisation}

We begin in Section \ref{sec:qwalk} by describing the main underlying quantum ingredient, the use of a quantum walk to detect a marked vertex in a tree. This algorithm is a special case of an algorithm described by Belovs~\cite{belovs13}. We then go on in Sections \ref{sec:finding} and \ref{sec:uniquemarked} to describe extensions to this algorithm to allow finding a marked vertex, and a faster runtime in the case where we know there is a unique marked vertex. Section \ref{sec:walkbacktrack} shows that the algorithm can be applied to accelerate backtracking algorithms for CSPs. Section \ref{sec:averagecase} discusses how to use the algorithm to obtain exponential reductions in expected runtime, while Section \ref{sec:barriers} concludes with a discussion of some ways in which the algorithm could be improved, and barriers to doing so.


\subsection{Preliminaries}

We will need the following tools, which have been used many times elsewhere in quantum algorithm design:

\begin{lem}[Effective spectral gap lemma~\cite{lee11}]
\label{lem:effsg}
Let $\Pi_A$ and $\Pi_B$ be projectors on the same Hilbert space, and set $R_A = 2\Pi_A - I$, $R_B = 2\Pi_B-I$. Let $P_\chi$ be the projector onto the span of the eigenvectors of $R_B R_A$ with eigenvalues $e^{2i \theta}$ such that $|\theta| \le \chi$. Then, for any vector $\ket{\psi}$ such that $\Pi_A \ket{\psi} = 0$, we have
\[ \|P_\chi \Pi_B \ket{\psi} \| \le \chi \|\ket{\psi} \|. \]
\end{lem}

\begin{thm}[Phase estimation~\cite{cleve98a,kitaev95}]
\label{thm:phaseest}
For every integer $s \ge 1$, and every unitary $U$ on $m$ qubits, there exists a uniformly generated quantum circuit $C$ such that $C$ acts on $m + s$ qubits and:
\begin{enumerate}
\item $C$ uses the controlled-$U$ operator $O(2^s)$ times, and contains $O(s^2)$ other gates.
\item For every eigenvector $\ket{\psi}$ of $U$ with eigenvalue 1, $C \ket{\psi}\ket{0^s} = \ket{\psi}\ket{0^s}$.
\item If $U \ket{\psi} = e^{2 i \theta} \ket{\psi}$, where $\theta \in (0,\pi)$, then $C \ket{\psi}\ket{0^s} = \ket{\psi}\ket{\omega}$, where $\ket{\omega}$ satisfies $|\ip{\omega}{0^s}|^2 = \sin^2(2^s \theta) / (2^{2s} \sin^2 \theta)$.
\item For any $\ket{\phi} \in (\C^2)^{\otimes m}$, expanded as $\ket{\phi} = \sum_k \lambda_k \ket{\psi_k}$, where $\ket{\psi_k}$ is an eigenvector of $U$ with eigenvalue $e^{2i\theta_k}$, then
\[ C\ket{\phi}\ket{0^s} = \sum_k \lambda_k \ket{\psi_k} \ket{\omega_k}, \]
where $\sum_{k: \theta_k \ge \epsilon} |\ip{\omega_k}{0^s}|^2 = O(1/(2^s \epsilon))$.
\end{enumerate}
Call $2^{-s}$ the precision of the circuit.
\end{thm}

Phase estimation is normally used to estimate eigenvalues of $U$ (hence its name); here, however, similarly to~\cite{magniez11} we will only need to apply it to distinguish the eigenvalue 1 from other eigenvalues. If the smallest nonzero phase is $\epsilon$, this can be done with $O(1/\epsilon)$ uses of controlled-$U$.

\begin{fact}[Close states and measurement outcomes, e.g.~\cite{bernstein97}]
\label{fact:closestates}
Let $\ket{\psi_1}$, $\ket{\psi_2}$ be quantum states satisfying $\| \ket{\psi_1} - \ket{\psi_2} \| = \epsilon$. Then the total variation distance between the two distributions on measurement outcomes obtained by measuring each state in the computational basis is at most $\epsilon$.
\end{fact}

(This fact is usually presented with $\epsilon$ replaced with $4\epsilon$~\cite{bernstein97}; the tighter constant stated here can easily be obtained by relating the fidelity of $\ket{\psi_1}$ and $\ket{\psi_2}$ to their trace distance, for example.)


\section{Quantum walks on trees}
\label{sec:qwalk}

We now describe a quantum algorithm for detecting a marked vertex in a tree. The algorithm is a special case of a beautiful connection between quantum walks and electrical circuits due to Belovs~\cite{belovs13} (see also~\cite{belovs13a}), which is a quantum analogue of a similar connection between random walks and electrical circuits~\cite{chandra97}. This is conceptually elegant and leads to a very concise proof of a previous result of Szegedy~\cite{szegedy04} on detecting marked elements using a quantum walk. Here we only use these ideas for the special case of trees and a quantum walk starting at the root. This will enable us to simplify some notation and, hopefully, make the algorithm more intuitive.

Consider a rooted tree with $T$ vertices, labelled $r,1,\dots,T-1$, with vertex $r$ being the root, where the distance from the root to any leaf is at most $n$. Assume for simplicity in what follows that the root is promised not to be marked. For each vertex $x$, let $\ell(x)$ be the distance of $x$ from the root. We assume throughout that, although we do not necessarily know the structure of $T$ in advance, we can determine $\ell(x)$ for any $x$. Let $A$ be the set of vertices an even distance from the root (including the root itself), and let $B$ be the set of vertices at an odd distance from the root. We write $x \rightarrow y$ to mean that $y$ is a child of $x$ in the tree. For each $x$, let $d_x$ be the degree of $x$ as a vertex in an undirected graph. Thus, for all $x \neq r$, $d_x = |\{y: x \rightarrow y\}| + 1$; and $d_r = |\{y:r\rightarrow y\}|$.

The quantum walk operates on the Hilbert space $\mathcal{H}$ spanned by $\{\ket{r}\} \cup \{ \ket{x} :  x \in \{1,\dots,T-1\} \}$, and starts in the state $\ket{r}$. Unlike many discrete-time quantum walk algorithms, it does not use a separate ``coin'' space. The walk is based on a set of diffusion operators $D_x$, where $D_x$ acts on the subspace $\mathcal{H}_x$ spanned by $\{\ket{x}\} \cup \{ \ket{y}: x \rightarrow y \}$. The diffusion operators are defined as follows:
\begin{itemize}
\item If $x$ is marked, then $D_x$ is the identity.
\item If $x$ is not marked, and $x \neq r$, then $D_x = I - 2 \proj{\psi_x}$, where
\[ \ket{\psi_x} = \frac{1}{\sqrt{d_x}} \left( \ket{x} + \sum_{y, x \rightarrow y} \ket{y} \right). \]
\item $D_r = I - 2 \proj{\psi_r}$, where
\[ \ket{\psi_r} = \frac{1}{\sqrt{1+ d_r n}} \left( \ket{r} + \sqrt{n} \sum_{y, r \rightarrow y} \ket{y}\right). \] 
\end{itemize}
Observe that $D_x$ can be implemented with only local knowledge, i.e.\ based only on whether $x$ is marked and the neighbourhood structure of $x$. A step of the walk consists of applying the operator $R_B R_A$, where $R_A = \bigoplus_{x \in A} D_x$ and $R_B = \proj{r} + \bigoplus_{x \in B} D_x$. An alternative way of viewing this process is as a quantum walk on the graph given by the {\em edges} of the tree, where we identify each vertex with the edge from its parent in the tree, and add an additional ``input'' edge into the root.

The algorithm for detecting a marked vertex is presented as Algorithm \ref{alg:detect}.

\boxalgm{alg:detect}{Detecting a marked vertex}{
{\bf Input:} Operators $R_A$, $R_B$, a failure probability $\delta$, upper bounds on the depth $n$ and the number of vertices $T$. Let $\beta, \gamma > 0$ be universal constants to be determined.
\begin{enumerate}
\item Repeat the following subroutine $K = \lceil \gamma \log (1/\delta) \rceil$ times:
\begin{enumerate}
\item Apply phase estimation to the operator $R_B R_A$ with precision $\beta/\sqrt{Tn}$.
\item If the eigenvalue is 1, accept; otherwise, reject.
\end{enumerate}
\item If the number of acceptances is at least $3K/8$, return ``marked vertex exists''; otherwise, return ``no marked vertex''.
\end{enumerate}
}

\begin{lem}[Special case of Belovs~\cite{belovs13}]
\label{lem:algdetect}
Algorithm \ref{alg:detect} uses $R_A$ and $R_B$ $O(\sqrt{Tn} \log(1/\delta))$ times. There exist universal constants $\beta$, $\gamma$ such that it fails with probability at most $\delta$.
\end{lem}

\begin{proof}
The complexity bound is immediate from Theorem \ref{thm:phaseest}. For the correctness proof, we first show that, if there is a marked vertex, then $\ket{r}$ is quite close to (a normalised version of) an eigenvector $\ket{\phi}$ of $R_B R_A$ with eigenvalue 1. Let $x_0$ be a marked vertex and set
\be \label{eq:phi} \ket{\phi} = \sqrt{n} \ket{r} + \sum_{x \neq r, x \leadsto x_0} (-1)^{\ell(x)} \ket{x}. \ee
Here $x \leadsto x_0$ denotes the vertices $x$ on the unique path from the root to $x_0$, including $x_0$ itself. To see that $\ket{\phi}$ is invariant under $R_B R_A$, first note that $\ket{\phi}$ is orthogonal to all states $\ket{\psi_x}$, where $x \neq r$ and $x$ is not marked. Indeed, any such state $\ket{\psi_x}$ either has uniform support on exactly 2 consecutive vertices $v$ in the path from $r$ to $x_0$, or is not supported on any vertices in this path. $\ket{\phi}$ is also orthogonal to $\ket{\psi_r}$ by direct calculation. We have
\[ \|\ket{\phi}\|^2 = n + \ell(x_0) \le 2 n. \]
Thus
\[ \frac{\ip{r}{\phi}}{\|\ket{\phi}\|} \ge \frac{1}{\sqrt{2}}. \]
Therefore, phase estimation returns the eigenvalue 1 with probability at least $1/2$. On the other hand, if there are no marked vertices, we consider the vector
\[ \ket{\eta} = \ket{r} + \sqrt{n} \sum_{x \neq r} \ket{x}. \]
Let $\Pi_A$ and $\Pi_B$ be projectors onto the invariant subspaces of $R_A$ and $R_B$. These spaces are spanned by vectors of the form $\ket{\psi_x^\perp}$ for $x \in A$, $x \in B$ respectively, where $\ket{\psi_x^\perp}$ is orthogonal to $\ket{\psi_x}$ and has support only on $\{\ket{x}\} \cup \{ \ket{y}: x \rightarrow y \}$; in addition to $\ket{r}$ in the case of $R_B$. 
On each subspace $\mathcal{H}_x$, $x \in A$, $\ket{\eta}$ is proportional to $\ket{\psi_x}$, so $\Pi_A \ket{\eta} = 0$. Similarly $\Pi_B \ket{\eta} = \ket{r}$. By the effective spectral gap lemma (Lemma~\ref{lem:effsg}), $\|P_\chi \ket{r} \| = \|P_\chi \Pi_B \ket{\eta}\| \le \chi \|\ket{\eta}\| \le \chi \sqrt{Tn}$. For small enough $\chi = \Omega(1/\sqrt{Tn})$, this is upper-bounded by $1/2$. By Theorem \ref{thm:phaseest}, there exists $\beta$ such that applying phase estimation to $R_B R_A$ with precision $\beta/\sqrt{Tn}$ returns the eigenvalue 1 with probability at most $1/4$.

Using a Chernoff bound, by repeating the subroutine $O(\log 1/\delta)$ times and returning ``marked vertex exists'' if the fraction of acceptances is greater than $3/8$, and ``no marked vertex'' otherwise, we obtain that the overall algorithm fails with probability at most $\delta$.
\end{proof}


\subsection{Finding a marked vertex}
\label{sec:finding}

From now on, we assume that the degree of every vertex in the tree is $O(1)$; this is not a significant restriction for the application to backtracking. For trees obeying this restriction we can use the detection algorithm as a subroutine to {\em find} a marked vertex efficiently, via binary search.

To find a marked vertex, we start by applying Algorithm \ref{alg:detect} to the entire tree. If it outputs ``marked vertex exists'', we apply the algorithm to the subtrees rooted at each child of the root in turn, to detect marked vertices within each subtree. Assuming the algorithm did not fail at any point, there must be a marked vertex in at least one subtree. We pick the root of one such subtree and check whether it is marked. If it is marked, we output its label and terminate; if it is not marked, we apply Algorithm \ref{alg:detect} to each of its children and repeat. This process continues until we have found a marked vertex. As there are at most $O(n)$ repetitions to reach a leaf and $O(1)$ subtrees are checked at each repetition, the time complexity of the algorithm is multiplied by a factor of $O(n)$. Note that, when we apply the algorithm to subtrees, we must leave the parameter $T$ unchanged; this is because the tree could be quite unbalanced, and a given subtree could contain many vertices.

We have thus far assumed that we know an upper bound on $T$ in advance. If we do not, we can repeat the whole search algorithm $O(\log T) = O(n)$ times, doubling a guess for $T$ each time (starting with $T=1$) until we either find a marked vertex, or the algorithm returns ``no marked vertex''. This exponential doubling does not affect the asymptotic runtime. If our guess for $T$ is too low, the correctness proof of Algorithm \ref{alg:detect} no longer holds, so the detection algorithm may claim that there is a marked vertex in a situation where there is actually no marked vertex. This may lead to the above binary search procedure returning an incorrect result. But we can deal with this situation by checking the final vertex returned by the search algorithm, and only terminating if it is marked; if it is not, we know that the search has failed, and continue doubling our guess for $T$. On the other hand, one can see from inspecting the proof of Lemma \ref{lem:algdetect} that, if there is a marked vertex, the phase estimation subroutine in Algorithm \ref{alg:detect} will accept with probability at least $1/2$ whether or not our guess for $T$ is large enough. Therefore, if there is a marked vertex, Algorithm \ref{alg:detect} will output that a marked vertex exists with probability at least $1-\delta$, for $\delta$ of our choice.

Using this procedure the total number of uses of Algorithm \ref{alg:detect} (with differing values of $T$) is $O(n^2)$, so in order for the whole algorithm to succeed with probability, say, $2/3$, it is sufficient to reduce the failure probability of each use of Algorithm \ref{alg:detect} to $O(1/n^2)$. This costs an additional time factor of $O(\log n)$ per use of the algorithm, giving a total runtime of  $O(\sqrt{T} n^{3/2} \log n)$. This can in turn be improved to an arbitrary failure probability $\delta > 0$ by taking $O(\log 1/\delta)$ repetitions, leading to an overall bound of time $O(\sqrt{T} n^{3/2} \log n \log(1/\delta))$.

Finally, we can find {\em all} marked vertices by simply repeating the algorithm, modifying the underlying oracle operator to strike out previously seen marked elements. If there are $k$ marked elements, the overall runtime is $O(k \sqrt{T} n^{3/2} \log n \log(k/\delta))$.


\subsection{Search with a unique marked element}
\label{sec:uniquemarked}

If we are promised that there exists a unique marked element in the tree, we can improve the above bounds by a factor of almost $n$. In general this improvement is not particularly large, as we usually have $T \gg n$; however, for some ``tall and thin'' trees it can be relatively significant. In particular, following this improvement we see that the complexity of the quantum algorithm for the search problem is never worse than the classical complexity $O(T)$, up to logarithmic factors.

We assume that there is a unique marked vertex $x_0$ and that $\ell(x_0) = n$. This second assumption is without loss of generality. We can determine $\ell(x_0)$ at the start of the algorithm by applying Algorithm \ref{alg:detect} to the subtree rooted at $r$ and of depth $i$, for differing values of $i$. That is, we only expand the tree up to depth $i$, and use binary search on $i \in \{1,\dots,n\}$ to find the minimal $i$ such that the tree of depth $i$ contains $x_0$. This needs $O(\log n)$ repetitions, so the complexity of this part is $O(\sqrt{Tn} \log n \log \log n)$, where the log log term comes from reducing the failure probability of Algorithm \ref{alg:detect} to $O(1/(\log n))$. Once $\ell(x_0)$ is determined, we henceforth only search within the tree of depth $\ell(x_0)$.

Let $\ket{\phi'} = \ket{\phi}/\|\ket{\phi}\|$, where the eigenvector $\ket{\phi}$ is defined in (\ref{eq:phi}), i.e.
\[ \ket{\phi'} = \frac{1}{\sqrt{2}} \ket{r} + \frac{1}{\sqrt{2n}} \sum_{x \neq r, x \leadsto x_0} (-1)^{\ell(x)} \ket{x}. \]
The starting point for the search algorithm is the observation\footnote{A similar observation was used in~\cite{wang13} to approximate effective resistances.} that $\ket{\phi'}$ encodes the entire path from $r$ to $x_0$. If we measure $\ket{\phi'}$, and do not receive outcome $r$, we receive a measurement outcome $y$ which is uniformly distributed on the path from $r$ to $x_0$. We can then repeat the algorithm on the subtree rooted at $y$, obtaining a new state of the form of $\ket{\phi'}$ for a smaller value of $n$. The expected number of measurements we would need to make to find $x_0$ is logarithmic in $n$ (rather than the bound of $O(n)$ which follows from the previous binary search algorithm).

We first bound the total number of quantum walk steps used to find $x_0$, given access to states of the form of $\ket{\phi'}$ for various subtrees. Let $C = O(1/\log n)$ be chosen such that Algorithm \ref{alg:detect} fails with probability at most $1/(4n)$ and uses at most $C \sqrt{Tn}$ steps. Given that $\ell(x_0) = n$, measuring a copy of $\ket{\phi'}$ will give a ``good'' outcome (which is not $r$) with probability $1/2$. The distance from the root of such an outcome is uniformly distributed. Considering only the good outcomes, the expected total number of steps $S_n$ to find $x_0$, given that $\ell(x_0) = n$, therefore satisfies
\[ S_n \le \frac{1}{n} \sum_{i=0}^{n-1} S_i + C \sqrt{Tn}. \]
We claim that $S_n = O(C\sqrt{Tn})$. The proof is by induction. First, $S_0 = 0$ as no quantum walk steps are made. Assume $S_i \le 4C \sqrt{Ti}$ for all $i < n$. Then
\[ S_n \le \frac{4C}{n} \sum_{i=0}^{n-1} \sqrt{Ti} + C \sqrt{Tn} \le 4C \sqrt{\frac{1}{n} \sum_{i=0}^{n-1} Ti } + C \sqrt{Tn} =\frac{4C}{\sqrt{2}}\sqrt{T}\sqrt{n-1} + C \sqrt{Tn} \le 4C\sqrt{Tn}, \]
where the second inequality is Jensen's inequality. As on average half the outcomes are good, the expected total number of steps is thus $O(\sqrt{Tn} \log n)$.

We can approximately produce $\ket{\phi'}$ by applying phase estimation to the operator $R_B R_A$, with input state $\ket{r}$. If we write
\[ \ket{r} = \frac{1}{\sqrt{2}} \ket{\phi'} + \frac{1}{\sqrt{2}} \ket{\phi^\perp}, \]
where $\ket{\phi^\perp}$ is normalised and orthogonal to $\ket{\phi}$, the result of applying phase estimation on $\ket{r}$ with $s$ ancilla qubits is a state of the form
\[  \frac{1}{\sqrt{2}} \ket{\phi'} \ket{0^s} + \frac{1}{\sqrt{2}} \sum_{k, \theta_k > 0} \lambda_k \ket{\psi_k} \ket{\omega_k}, \]
where $\ket{\psi_k}$ is an eigenvector of $R_B R_A$ with eigenvalue $e^{2i\theta_k}$. Write each $\ket{\omega_k}$ as $\ket{\omega_k} = \mu_k \ket{0^s} + \ket{\omega_k'}$ for some subnormalised vectors $\ket{\omega'_k}$ orthogonal to $\ket{0^s}$. If we obtain outcome $\ket{0^s}$ when we measure the second register, the first register collapses to
\[ \ket{\widetilde{\phi}'} = \frac{1}{\sqrt{1+\sum_{k,\theta_k > 0} |\lambda_k \mu_k|^2}} \left( \ket{\phi'} + \sum_{k,\theta_k>0} \lambda_k \mu_k \ket{\psi_k} \right). \]
To bound the distance between $\ket{\widetilde{\phi}'}$ and the desired state $\ket{\phi'}$, we split the sum into two parts. For any $\epsilon > 0$, via Theorem \ref{thm:phaseest} we have
\[ \sum_{k,\theta_k \ge \epsilon} |\lambda_k \mu_k|^2 \le \sum_{k,\theta_k \ge \epsilon} |\mu_k|^2 = O(1 / (2^s \epsilon)). \] 
On the other hand, we prove the following technical claim in Appendix \ref{app:lowepsbound}. Recall that $P_\epsilon$ is the projector onto the span of the eigenvectors of $R_B R_A$ with eigenvalues $e^{2i \theta}$ such that $|\theta| \le \epsilon$.

\begin{lem}
\label{lem:lowepsbound}
$\|P_\epsilon \ket{\phi^\perp} \| = O(\epsilon \sqrt{Tn})$.
\end{lem}

Given Lemma \ref{lem:lowepsbound}, we have
\[ \sum_{k,0<\theta_k \le \epsilon} |\lambda_k \mu_k|^2 \le \sum_{k,0<\theta_k \le \epsilon} |\lambda_k|^2 = \|P_\epsilon \ket{\phi^\perp}\|^2 = O(\epsilon^2 Tn). \] 
Fixing an accuracy $\delta$ and taking $\epsilon = \Theta(\delta/\sqrt{Tn})$, $2^s = O(\sqrt{Tn}/\delta^3)$, we have $\| \ket{\widetilde{\phi}'} - \ket{\phi'} \| = O(\delta)$. By Fact \ref{fact:closestates}, measuring $\ket{\widetilde{\phi}'}$ in the computational basis is indistinguishable from measuring $\ket{\phi'}$, except with probability $O(\delta)$. If we take $\delta = O(1/\log n)$, the algorithm does not notice the difference on any of the $O(\log n)$ states used, with probability $\Omega(1)$. The overall complexity of the algorithm is therefore $O(\sqrt{Tn} \log^3 n)$\footnote{One way to improve the polylogarithmic factors in this complexity could be to reweight the tree such that the eigenvector of $R_B R_A$ with eigenvalue 1 has more weight on $x_0$ (Alexander Belov, personal communication).}. As before, the failure probability can be made arbitrarily small via repetition.

In Section \ref{sec:barriers} we discuss some barriers to improving the complexity and applicability of these algorithms.


\section{From quantum walks on trees to accelerating backtracking}
\label{sec:walkbacktrack}

To complete the proofs of Theorems \ref{thm:detectbound} and \ref{thm:findbound}, we now verify that Algorithm \ref{alg:detect}  can be applied to search in the tree defined by a backtracking algorithm. In order to do this, it is sufficient to define a suitable efficient mapping between partial assignments and vertices in a tree, and to implement the operators $R_A$ and $R_B$ appropriately and efficiently. As the quantum walk subroutines assume that the root of the tree is not marked, the first step of the algorithm is to check whether $P(\ast^n)$ is true. If so, the algorithm immediately returns ``true''; if not, it runs Algorithm \ref{alg:detect} on a graph defined as follows.

The current state of the backtracking algorithm is represented by a vertex in a rooted tree labelled with a sequence of the form $(i_1,v_1),\dots,(i_\ell,v_\ell)$, for $1 \le \ell \le n$. The sequence corresponds to a partial assignment $x \in \mathcal{D}$ where we assign $x_{i_k} = v_k$ for $k=1,\dots,\ell$, and $x_j = \ast$ for all other indices $j$. The tree only contains vertices corresponding to valid partial assignments. Each vertex except for the root (which is labelled with the empty sequence) is connected to its parent, the vertex labelled with $(i_1,v_1),\dots,(i_{\ell-1},v_{\ell-1})$. It is also connected to all vertices of the form $(i_1,v_1),\dots,(i_\ell,v_\ell),(j,w)$, where $j = h((i_1,v_1),\dots,(i_\ell,v_\ell))$, $w \in [d]$, and $P((i_1,v_1),\dots,(i_\ell,v_\ell),(j,w))$ is not false. That is, all vertices corresponding to valid partial assignments which extend the current partial assigment by assigning a value to the variable whose index is given by $h$. It is convenient to assume that the predicate $P$ and the heuristic $h$ take as input a string of (index, value) pairs which describe value assignments to variables, rather than an element of $\mathcal{D}$; if not, converting between these representations can be done in time $O(n)$. We will also assume that, for all complete assignments, the predicate returns either true or false (as it should do).

The algorithm takes place within the Hilbert space $\mathcal{H}^{(n)} = \C^{n+1} \otimes (\C^{n+1} \otimes \C^{d+1})^{\otimes n}$ together with an ancilla space. Each basis vector within $\mathcal{H}^{(n)}$ represents a partial assignment described by a sequence as above. The first register stores a level $\ell$ between 0 and $n$, representing the length of the sequence (the number of non-$\ast$'s in the assignment). Each of the next $\ell$ registers stores a pair $(i_k,v_k)$ giving the index of a variable (an integer between 1 and $n$) and the assignment to that variable (an integer between 0 and $d-1$). Except during updates to the state, the remaining $n-\ell$ registers all contain the pair $(0,\ast)$. The algorithm can easily be modified to use qubits if desired, rather than systems with dimension $n+1$ and $d+1$, by encoding each subsystem in $O(\log n + \log d)$ qubits.

Let $U_{\alpha,S}$, for $S \subseteq [d]$ and $\alpha \in \R$, act on $\C^{d+1}$ with basis $\{ \ket{\ast},\ket{0},\dots,\ket{d-1} \}$ by mapping $\ket{\ast} \mapsto \ket{\phi_{\alpha,S}}$, where
\[ \ket{\phi_{\alpha,S}} := \frac{1}{\sqrt{\alpha |S|+1}} \left(\ket{\ast} + \sqrt{\alpha} \sum_{i \in S} \ket{i} \right). \]
We assume that, for any subset $S \subseteq [d]$ and any fixed $\alpha \in \R$, we can perform $U_{\alpha,S}$ and its inverse in time $O(1)$ each. (Dependent on the gate set being used, we may not be able to implement $U_{\alpha,S}$ exactly. However, for any universal gate set we can implement it up to accuracy $1-\epsilon$ in time $\poly \log (1/\epsilon)$; this will multiply the runtime of the overall algorithm by at most a polylogarithmic factor.) By applying $U_{\alpha,S}$ and its inverse we can perform the operation $I - 2 \proj{\phi_{\alpha,S}}$.

In order to use Algorithm \ref{alg:detect}, we need to implement the operators $R_A$ and $R_B$. The implementation of $R_A$ using $I - 2 \proj{\phi_{\alpha,S}}$, $P$ and $h$ is described in Algorithm \ref{alg:implra} above. $R_B$ is similar, except that: step 1 is replaced with the check ``If $P(x)$ is true or $\ell=0$, return''; ``odd'' is replaced with ``even'' in steps 2 and 8; and the check ``If $\ell=0$'' is removed from step 6. The first of these changes is because $R_B$ should leave the root of the tree invariant; and the last is because $\ell$ is always odd at that point in the modified algorithm, so the check is unnecessary.

\boxalgm{alg:implra}{Implementation of the operator $R_A$}{
{\bf Input:} A basis state $\ket{\ell}\ket{(i_1,v_1)}\dots \ket{(i_n,v_n)} \in \mathcal{H}^{(n)}$ corresponding to a partial assignment
$x_{i_1} = v_1, \dots, x_{i_\ell} = v_\ell$.  
Ancilla registers $\mathcal{H}_{\text{anc}}$, $\mathcal{H}_{\text{next}}$, $\mathcal{H}_{\text{children}}$, storing a tuple $(a,j,S)$, where $a \in \{\ast\} \cup [d]$, $j \in \{0,\dots,n\}$, $S \subseteq [d]$, initialised to $a=\ast$, $j=0$, $S=\emptyset$.

\begin{enumerate}
\item If $P(x)$ is true, return.
\item If $\ell$ is odd, subtract $h((i_1,v_1),\dots,(i_{\ell-1},v_{\ell-1}))$ from $i_\ell$ and swap $a$ with $v_{\ell}$.
\item If $a \neq \ast$, subtract 1 from $\ell$. (Now $\ell$ is even and $(i_{\ell+1},v_{\ell+1}) = (0,\ast)$.)
\item Add $h((i_1,v_1),\dots,(i_\ell,v_\ell))$ to $j$.
\item For each $w \in [d]$:
\begin{enumerate}
\item If $P((i_1,v_1),\dots,(i_\ell,v_\ell),(j,w))$ is not false, set $S = S \cup \{w\}$.
\end{enumerate}
\item If $\ell = 0$, perform the operation $I - 2 \proj{\phi_{n,S}}$ on $\mathcal{H}_{\text{anc}}$. Otherwise, perform the operation $I - 2 \proj{\phi_{1,S}}$ on $\mathcal{H}_{\text{anc}}$.
\item Uncompute $S$ and $j$ by reversing steps 5 and 4.
\item If $a \neq \ast$, add 1 to $\ell$. If $\ell$ is now odd, add $h((i_1,v_1),\dots,(i_{\ell-1},v_{\ell-1}))$ to $i_\ell$ and swap $v_\ell$ with $a$. (Now $a = \ast$ again.) 
\end{enumerate}
}

We now argue that Algorithm \ref{alg:implra} correctly implements $R_A$. Write $x = (i_1,v_1),\dots,(i_\ell,v_\ell)$ for the partial assignment passed to the algorithm, and write $x' = (i_1,v_1),\dots,(i_{\ell-1},v_{\ell-1})$ for the parent partial assignment in the tree. The goal of the algorithm is to implement the operator $\bigoplus_{x \in A} D_x$ defined in Section \ref{sec:qwalk}. For each $x \in A$, $D_x$ only acts on the subspace corresponding to $x$ and its children. To implement $D_x$, it is therefore sufficient to map the basis state corresponding to $(i_1,v_1),\dots,(i_\ell,v_\ell)$, and all the basis states corresponding to $(i_1,v_1),\dots,(i_\ell,v_\ell),(j,w)$ for $w \in [d]$, where $j = h((i_1,v_1),\dots,(i_\ell,v_\ell))$ and $\ell$ is even, to a $(d+1)$-dimensional subspace on which the children of $x$ can be mixed over using $U_{\alpha,S}$, and then returning to the original subspace. This is precisely what Algorithm \ref{alg:implra} does.

In more detail, the algorithm performs the following steps. First, it does nothing when $x$ is marked, corresponding to the definition of $D_x$. If $x$ is not marked, the behaviour depends on whether $\ell$ is even (corresponding to $x \in A$) or $\ell$ is odd (corresponding to $x \in B$). Define $y$ by setting $y=x$ if $x\in A$, and $y=x'$ if $x \in B$. Then the algorithm implements an inversion about $\ket{\psi_y}$, which is split into three subparts:
\begin{itemize}
\item Steps 2-3: Perform a map of the form $\ket{x} \mapsto \ket{y}\ket{\ast}$ for $x \in A$, and $\ket{x} \mapsto \ket{y}\ket{w}$ for $x \in B$, where $w$ is the value of $x$ at the $h(x')$'th position, i.e.\ the most recent variable assignment that was made by the backtracking algorithm. 
\item Steps 4-5: Determine the children of $y$.
\item Step 6: Perform the operation $I - 2\proj{\psi_y}$ using the knowledge of the children of $y$.
\item Steps 7-8: Uncompute junk and reverse the first map.
\end{itemize}
It can be verified that the algorithm implements the desired behaviour for all basis state inputs $\ket{\ell}\ket{(i_1,v_1)}\dots \ket{(i_n,v_n)}$ such that $(i_1,v_1),\dots,(i_\ell,v_\ell)$ is a valid path in the backtracking tree; we omit the routine details. As the algorithm implements the operation $R_A = \bigoplus_{x \in A} D_x$ unitarily for all basis states $\ket{x}$, it also implements $R_A$ correctly for all superpositions of basis states. Together with the similar implementation of $R_B$, this is enough to implement Algorithm \ref{alg:detect}. For each use of $R_A$ and $R_B$ the algorithm uses $O(1)$ auxiliary operations as claimed.


\section{From quadratic speedups to exponential speedups}
\label{sec:averagecase}

In this section we show that it is possible to leverage the speedup achieved by the quantum backtracking algorithm to obtain much more significant speedups over classical algorithms -- but in a non-standard, average-case setting.

For any (classical or quantum) algorithm $\mathcal{A}$, let $T_\mathcal{A}(X)$ denote the expected runtime of $\mathcal{A}$ on input $X$. Let $\mathcal{P}$ be a distribution on inputs $X$. Imagine we have a quantum algorithm $\mathcal{Q}$ and a classical algorithm $\mathcal{C}$ such that $T_\mathcal{Q}(X) \approx \sqrt{ T_\mathcal{C}(X)}$ for all $X$. This is the case for the quantum algorithms presented here, where for CSPs on $n$ variables we have $T_\mathcal{Q}(X) \le \sqrt{T_\mathcal{C}(X)} \poly(n)$. Then, by Jensen's inequality, we have
\[ \E_{X \sim \mathcal{P}}[T_{\mathcal{Q}}(X)] \le \sqrt{\E_{X \sim \mathcal{P}}[T_{\mathcal{Q}}(X)^2]}  \lesssim \sqrt{\E_{X \sim \mathcal{P}}[T_{\mathcal{C}}(X)]}. \]
However, dependent on the distribution $\mathcal{P}$, taking the average in this way can sometimes amplify the separation to become much greater than quadratic, and even sometimes exponential or super-exponential. This point was noted in the context of quantum query complexity by Ambainis and de Wolf~\cite{ambainis01}, who gave several examples of super-polynomial average-case quantum speedups for the computation of total functions, and later by Montanaro~\cite{montanaro10}, who showed that even the unstructured search problem with a unique marked element, with power-law distributions on the position of this marked element, can display this behaviour.

One very simple example of this phenomenon is the following separation. Let $\mathcal{C}$ be a classical algorithm for Circuit SAT. Assume that, for each integer $n$, there exists an instance of Circuit SAT on $n$ variables such that $\mathcal{C}$ has runtime $\Omega(2^n)$ (this is the case for the best classical algorithms at present~\cite{williams13}). Also let $\mathcal{Q}$ be a quantum algorithm which solves Circuit SAT using Grover's algorithm, using time $O(2^{n/2} \poly(n))$ on an input of size $n$. Finally, let $\mathcal{P}_n$ be the following distribution on instances with $n$ variables: with probability $p$, return a hard instance of size $n$; with probability $1-p$, return a trivial instance. Then
\[ \E_{X \sim \mathcal{P}_n}[T_{\mathcal{C}}(X)] = \Omega(p 2^n), \;\;\;\; \E_{X \sim \mathcal{P}_n}[T_{\mathcal{Q}}(X)] = O(p 2^{n/2} \poly(n)). \]
If we take $p = 2^{-n/2}$, the separation between these two quantities is exponential.

However, this is clearly a rather contrived distribution on the inputs. One might hope to find some problem, together with a more natural distribution on the inputs, which allows a similar exponential separation to be proven. The quantum backtracking algorithm allows one to find separations of this form, given a backtracking algorithm with a suitable distribution of runtimes. Indeed, imagine we have a family of CSPs and a distribution $\mathcal{P}_n$ on problems of size $n$ such that with high probability the problem has $O(1)$ solutions. Further imagine that we have a deterministic classical backtracking algorithm whose backtracking tree contains $T(X)$ vertices on input $X$, such that $\Pr_{\mathcal{P}_n}[T(X) = t] \le C t^{\beta}$ for all $t$ and some constants $C$ and $\beta$. In addition, assume that $\Pr_{\mathcal{P}_n}[T(X) = t] \ge D t^{\beta}$, for some constant $D$, for $M$ different values $t$. Here $M$ is some large integer which we think of as being exponentially large in $n$. Then
\[ \E_{X \sim \mathcal{P}_n}[T(X)] \ge \sum_{t=1}^{M} D t^{\beta} \cdot t = \Omega(M^{\beta+2}). \]
For $\beta > -2$, this quantity is exponentially large. However, if $-2 < \beta < -3/2$, the quantum backtracking algorithm described above uses an average of
\[ \E_{X \sim \mathcal{P}_n}[O(\sqrt{T(X)}\poly(n))] \le \sum_{t \ge 1} O(\sqrt{t} \cdot t^{\beta} \poly(n)) = \poly(n) \]
quantum walk steps. If each step requires time $\poly(n)$ (to evaluate the predicate $P$ and the heuristic $h$) we have obtained an exponential reduction in expected runtime.

We therefore see that a ``power law'' tail of the distribution $\mathcal{P}_n$ of the form $p_t = \Pr_{\mathcal{P}_n}[T(X) = t] \sim t^{\beta}$, for a suitable value of $\beta$, gives us an exponential separation. There is substantial empirical evidence, and some analytical evidence, that such power law, or ``heavy'', tails can occur in both random and real-world instances of CSPs; for a survey, see~\cite{gomes06}. For example, consider the case of graph $k$-colouring on random graphs with $n$ vertices, where each edge is present with probability $\Theta(1/n)$. Hogg and Williams~\cite{hogg94} observed that a natural backtracking algorithm seemed to have a power-law distribution of its runtimes. Later work by Jia and Moore~\cite{jia04} provided some analytical justification for this, and additional experiments, which together suggest that for 3-colouring the distribution is of the form $p_t \sim t^{-1}$.

However, there are several reasons why it is unclear that this phenomenon could lead to exponential separations between quantum and classical expected runtimes. First, there is some evidence that some apparently heavy-tailed behaviour may in fact be due to finite-size effects~\cite{cocco02,cocco05}. Second, one reason for a skewed runtime distribution could be that, on satisfiable instances, the backtracking algorithm sometimes gets lucky and happens to find a satisfying assignment early on, after which it terminates. The runtime of the quantum algorithm described here depends on the size of the whole tree and hence will not correspond to the square root of the classical runtime in this case. Indeed, runtime distributions on unsatisfiable instances do not seem to display the same heavy-tailed behaviour~\cite{rish97,gomes04}.

Third, in many cases power-law behaviour is observed when a randomised backtracking algorithm is run on a {\em single} instance. That is, when the choices of branching variables made by the algorithm are random and we consider the distribution of the runtimes $T(r)$ over the choice of random seed $r$. Algorithmic randomness of this form (as opposed to picking the input instance at random) is not suitable for obtaining an exponential quantum-classical separation using the quantum backtracking algorithm. This is because, if the quantum backtracking algorithm's expected runtime over $r$ is at most $R$, for some $R$, we have $R = \Omega(\E_r[\sqrt{T(r)}])$. So, by Markov's inequality, $T(r) = O(R^2)$ with, say, 99\% probability. Therefore, if we stop the classical algorithm after time $O(R^2)$, it will succeed with probability $\Omega(1)$.

For these reasons, we consider random instances of CSPs produced not just by using a fixed density of constraints, but by taking a distribution over different constraint densities. This enables us to find relatively natural input distributions under which the expected runtime of the quantum backtracking algorithm is exponentially faster.


\subsection{Expected runtime bounds}

There is now a substantial body of work proving bounds on the expected runtime of DPLL-type algorithms for random $k$-SAT. For example, consider $k=3$ and instances consisting of $m = \alpha n$ uniformly random clauses. For $\alpha \gtrsim 4.3$, there is strong evidence that such instances very rarely have a solution~\cite{selman96}, so the task of the algorithm is usually to prove unsatisfiability. Beame et al.~\cite{beame02} have shown that, for a simple DPLL variant (known as ordered DLL) the runtime is $2^{\Theta(n/\alpha)}$ with probability $1-o(1)$. Cocco and Monasson~\cite{cocco01,cocco01a} have used statistical physics techniques to even determine (non-rigorously) the constant in the exponent. In particular, they argue that, for large $\alpha$, the runtime is approximately $2^{0.292n/ \alpha}$.


Sometimes one can prove such tight bounds rigorously. For example, consider the following very simple backtracking algorithm, which fits within the framework of Algorithm \ref{alg:backtrack}. Fix an ordering of the variables from 1 to $n$. Then the heuristic $h$ returns the lowest index of a variable which has not yet been assigned a value. Call this algorithm Na\"iveBt. Then the following result holds:

\begin{prop}
\label{prop:exp}
The expected number of vertices $E$ in the backtracking tree of the Na\"iveBt algorithm when applied to a random $k$-SAT instance on $n$ variables, with $m = \alpha n$ uniformly random clauses, for $1 \le \alpha \le n^{k-1}$, satisfies
\[ 2^{C' n} \le E \le O(n 2^{C n}), \]
where $C$ and $C'$ depend only on $\alpha$ and $k$. For $k=3$, $C \le 0.907/\sqrt{\alpha}$, $C' \ge 0.906/\sqrt{\alpha} - 0.142/\alpha^2$.
\end{prop}

\begin{proof}
See Appendix \ref{app:ordereddll}.
\end{proof}

Similar analyses to Proposition \ref{prop:exp} have been carried out many times in the literature, albeit often for slightly different models (e.g.~\cite{brown81,purdom97}).

Let $C_n$ denote the expected runtime of the Na\"iveBt algorithm applied to 3-SAT instances on $n$ variables, where the expectation is taken with respect to a distribution over numbers of constraints $m$. If the probability that we have $m$ constraints is $p_m$, then by Proposition \ref{prop:exp}
\[ C_n \ge \sum_m p_m 2^{n(0.906 \sqrt{n/m} - 0.142 (n/m)^2) }. \]
On the other hand, up to a $\poly(n)$ factor, on any instance the runtime of the quantum backtracking algorithm is at most the square root of the classical one, multiplied by the number of solutions (if there are any). We consider a distribution $p_m$ which is only supported on values of $m$ such that the probability that a random CSP with $m$ constraints has a solution is $O(2^{-n})$. For such values of $m$, we can ignore the additional cost for finding all the solutions. For 3-SAT, this is true for $m > 16n/(\ln 2)$, for example (see Appendix \ref{app:ordereddll}).

Letting $Q_n$ denote the expected runtime of the quantum algorithm applied to Na\"iveBt, we therefore have
\[ Q_n \le \sum_m p_m 2^{0.454 n^{3/2}/\sqrt{m}} \poly(n). \]
Consider the distribution $p_m \propto 2^{-0.454 n^{3/2}/\sqrt{m}}$ for $16n/(\ln 2) < m \le n^3$. Then
\[ Q_n = \poly(n),\;\;\;\; C_n = \Omega( \sum_{m > 16n/(\ln 2)} 2^{(0.906-0.454) n^{3/2}/\sqrt{m} - 0.142 (n/m)^2 } ) = \Omega(2^{0.094n}),  \]
where the asymptotic bound follows from inserting $m = \lceil 16n/(\ln 2) \rceil$. We therefore have an average-case exponential separation between the quantum and classical complexities of 3-SAT under this distribution; and, indeed, for various other distributions of the form $p_m \propto 2^{-Cn^{3/2}/\sqrt{m}}$. While this family of distributions is arguably less contrived than the Circuit SAT example given above (for example, the number of variables is fixed; only the number of clauses varies), it still appears somewhat unnatural. It seems to be an interesting question to determine more natural input distributions which also lead to exponential quantum speedups.


\section{Improving the quantum walk algorithm?}
\label{sec:barriers}

We finish by addressing the question of how tight the bounds are which we have obtained on quantum search in trees. It is clear that, given a tree with $T$ vertices, we must have a lower bound of the form $\Omega(\sqrt{T})$ for finding a marked vertex (otherwise, we could use the algorithm to solve the unstructured search problem on $T$ elements using $o(\sqrt{T})$ quantum queries, which is impossible~\cite{bennett97}). There are several plausible ways in which the complexity of the algorithm presented here could be improved to get closer to this bound. However, there appear to be some challenges to doing so in each of these cases.

\begin{enumerate}
\item Reduction of the dependence on the depth $n$. It is easy to see that, if we would like to apply the quantum backtracking algorithm to general trees, there must be some dependence on the depth in the runtime. 
Indeed, consider a path on $T$ vertices, which has depth $T-1$. Then, if the marked vertex is the last one in the path, we require $\Omega(T)$ steps to find it. More generally, it was shown by Aaronson and Ambainis~\cite{aaronson05} that for each pair $T$ and $n$, there is a tree containing $T$ vertices and with depth $O(n)$ such that determining the existence of a marked vertex requires $\Omega(\sqrt{Tn})$ queries. This holds even if we know the tree in advance and are allowed to perform arbitrary ``local'' operations to search within it.

\item Reduction of the overhead for searching with multiple marked vertices. It would be interesting to determine whether the search algorithm in Section \ref{sec:uniquemarked} could be generalised to work with a similar efficiency for an arbitrary number of marked vertices. The question of when one can convert a quantum walk speedup for detecting a marked element to a speedup for finding a marked element has been studied previously. But while it was shown by Szegedy~\cite{szegedy04} that the time to detect a marked element using a quantum walk is at most the square root of the classical hitting time, it is not  known whether the time to find a marked element has the same scaling in general.

Indeed, Krovi et al.~\cite{krovi15} have described a way (generalising previous results of~\cite{tulsi08,magniez12}) to modify the original quantum walk approach of Szegedy to obtain a quadratic speedup for the search problem in the case where there is a unique marked element. However, if there is more than one marked element, the runtime of their algorithm scales with a quantity they call the extended hitting time, which may be larger than the hitting time. In any case, all these algorithms assume that the graph is known in advance and the initial state of the quantum walk algorithm corresponds to the stationary distribution of the random walk. Neither of these assumptions applies here.


\item Reduction of the dependence on $k$ to find one, or all, of $k$ marked vertices. For the unstructured search problem with $k$ marked elements out of $T$, Grover's algorithm can find a marked element using $O(\sqrt{T / k})$ queries, which implies an algorithm which finds all marked elements in $O(\sqrt{Tk})$ queries. It would be natural to hope for a bound of a similar form for quantum search on trees, e.g.~$O(\sqrt{Tn/k})$ to find a marked vertex and $O(\sqrt{Tnk})$ to find all $k$ of them. Unfortunately, it is far from clear that this can be achieved.

Indeed, consider the following argument due to Alexander Belov. Imagine we have access to an algorithm $\mathcal{A}$ which finds one of $k > 1$ marked vertices using $o(\sqrt{Tn})$ queries, and consider an arbitrary tree containing one marked leaf $\ell_0$. Modify the tree by attaching a subtree of depth $O(\log k)$ below that leaf containing $k$ vertices, all of which are marked and are labelled such that $\ell_0$ can be determined from their labels. Then, using $\mathcal{A}$, we can find one of these vertices using $o(\sqrt{Tn})$ queries. Finding such a vertex enables us to find $\ell_0$ with no additional queries, contradicting the aforementioned $\Omega(\sqrt{Tn})$ lower bound~\cite{aaronson05}. However, this argument does not rule out the possibility that some other approach could find all $k$ marked vertices in, for example, $O(\sqrt{Tnk})$ time.

\end{enumerate}
One other way in which it might be possible to improve the quantum backtracking algorithm is in situations where the classical backtracking algorithm is lucky and finds a solution without exploring the whole tree. For such instances the quantum algorithm, which is forced to explore the whole tree, may not outperform the classical algorithm. It might be possible to improve the performance of the quantum algorithm in this situation by biasing it to prefer to explore the parts of the tree visited by the classical algorithm earlier on.


\subsection*{Acknowledgements}

I would like to thank Alexander Belov and Aram Harrow for helpful comments on previous versions of this paper. This work was supported by EPSRC Early Career Fellowship EP/L021005/1.


\appendix

\section{Proof of technical claim for search with one marked element}
\label{app:lowepsbound}

In this appendix, we prove the following claim from Section \ref{sec:uniquemarked}:

\begin{replem}{lem:lowepsbound}
$\|P_\chi \ket{\phi^\perp} \| = O(\chi \sqrt{Tn})$.
\end{replem}

Let $x_0$ be the unique marked vertex, assuming for simplicity in the proof (as justified in Section \ref{sec:uniquemarked}) that $\ell(x_0)=n$, and hence that $x_0$ is a leaf in the tree. We can write
\beas \ket{\phi^\perp} &=& \sqrt{2} \ket{r} - \ket{\phi'} = \sqrt{2} \ket{r} - \frac{1}{\sqrt{2n}} \left( \sqrt{n} \ket{r} + \sum_{x \neq r, x \leadsto x_0} (-1)^{\ell(x)} \ket{x} \right)\\
&=& \frac{1}{\sqrt{2}} \ket{r} - \frac{1}{\sqrt{2n}} \sum_{x \neq r, x \leadsto x_0} (-1)^{\ell(x)} \ket{x}. 
\eeas
Recall that $\Pi_A$ and $\Pi_B$ are projectors onto the invariant subspaces of $R_A$ and $R_B$. The invariant subspace of $R_A$ is spanned by vectors of the form $\ket{\psi_x^\perp}$ for each vertex $x \in A$, and if $x_0 \in A$, in addition the vector $\ket{\psi_{x_0}}$. The invariant subspace of $R_B$ is similar (replacing $A$ with $B$) but also contains $\ket{r}$. Here $\ip{\psi_x}{\psi_x^\perp} = 0$ and $\ket{\psi_x^\perp}$ has support only on $\{\ket{x}\} \cup \{ \ket{y}: x \rightarrow y \}$. In order to apply the effective spectral gap lemma, we determine a vector $\ket{\xi}$ such that $\Pi_A \ket{\xi} = 0$ and $\Pi_B \ket{\xi} = \ket{\phi^\perp}$.

First assume $x_0 \in B$. We will take $\ket{\xi}$ to be a linear combination of vectors $\ket{\psi_x}$ for $x \in A$. Then the first of these two constraints is immediately satisfied. The second will be satisfied if, for a set of vectors $\ket{\zeta}$ which span the invariant subspace of $R_B$, i.e.
\[ \ket{\zeta} \in \{ \ket{r}, \ket{\psi_{x_0}} \} \cup \{ \ket{\psi^\perp_x}: \ip{\psi^\perp_x}{\psi_x} = 0, x \in B \}, \]
we have $\ip{\zeta}{\xi} = \ip{\zeta}{\phi^\perp}$. To compute the required inner products, first observe that $\ket{\psi^\perp_x}$ only has support on $x$ and its children, so for all $x$ not on the path from $r$ to $x_0$, $\ip{\psi^\perp_x}{\phi^\perp} = 0$. On the other hand, for each $x \in B$ such that $x \leadsto x_0$, define a basis for the space $\spann \{ \ket{\psi^\perp_x}: \ip{\psi^\perp_x}{\psi_x} = 0 \}$ by fixing the vectors
\[ \ket{\psi^\perp_{x,i}} = -\ket{x} + (d_x-1) \ket{N_i(x)} - \sum_{j \neq i} \ket{N_j(x)}, \]
where $N_i(x)$ denotes the $i$'th child of $x$, recalling that $d_x$ denotes the degree of $x$. We have
\be \label{eq:perpeqns} \ip{\psi^\perp_{x,i}}{\phi^\perp} = \begin{cases} -d_x/\sqrt{2n} & \text{if } i = i_0 \\ 0 & \text{otherwise} \end{cases} \ee
where $i_0$ denotes the unique child of $x$ on the path to $x_0$. 

We now find a vector $\ket{\xi} = \sum_x \alpha_x \ket{x}$ satisfying the above constraints. First, we require $\alpha_r = \ip{r}{\xi} = \ip{r}{\phi^\perp} = 1/\sqrt{2}$ and $\alpha_{x_0} = \ip{x_0}{\phi^\perp} = 1/\sqrt{2n}$. For each $x$, let $x'$ denote the parent of $x$ in the tree. For $\ket{\xi}$ to be a linear combination of vectors $\ket{\psi_x}$, $x \in A$, it is necessary and sufficient that $\alpha_x = \alpha_{x'}$ for all $x \in B$; except in the case $\ell(x) = 1$, where we require $\alpha_x = \sqrt{n}\,\alpha_r$. We in addition need $\alpha_x = \alpha_{x'}$ for all $x \neq r \in A$ such that $x'$ is not on the path to $x_0$, in order that $\ip{\psi^\perp_{x'}}{\xi} = \ip{\psi^\perp_{x'}}{\phi^\perp} = 0$. For each child $y$ of $x \in B$, set $\alpha_y = \gamma$ if $y \not\leadsto x_0$, and $\alpha_y = \delta$ if $y \leadsto x_0$. Then from (\ref{eq:perpeqns}) we have the final constraints that
%
\[ -\alpha_x + 2\gamma - \delta = 0 \text{ if } y \not\leadsto x_0, \]
\[ -\alpha_x - (d_x-2)\gamma + (d_x-1)\delta = -\frac{d_x}{\sqrt{2n}} \text{ if } y \leadsto x_0. \]
These equations have unique solution $\gamma = \alpha_x - 1/\sqrt{2n}$, $\delta = \alpha_x - \sqrt{2/n}$. This now uniquely defines all coefficients $\alpha_x$. In particular, observe that
\[ \alpha_{x_0} = \sqrt{\frac{n}{2}} - \left(\frac{n-1}{2}\right)\sqrt{\frac{2}{n}} = \frac{1}{\sqrt{2n}} \]
as required.

This constructs $\ket{\xi}$ in the case where $x_0 \in B$. If instead $x_0 \in A$, the procedure is similar. Now $\ket{\psi_{x_0}}$ is not in the invariant subspace of $R_B$ (which only makes it easier to satisfy the inner product constraints), but also $\ket{\xi}$ must be a linear combination of vectors $\ket{\psi_x}$ corresponding only to {\em unmarked} vertices $x \in A$. This new constraint implies that now $\alpha_{x_0} = 0$. But following the above procedure now gives
\[ \alpha_{x_0} = \sqrt{\frac{n}{2}} - \frac{n}{2}\sqrt{\frac{2}{n}} = 0 \]
as required. In either case, for all $x$, we have $|\alpha_x| \le \sqrt{n/2}$. So $\|\ket{\xi}\| = O(\sqrt{Tn})$ and hence, by the effective spectral gap lemma (Lemma \ref{lem:effsg}), $\|P_\chi \ket{\phi^\perp}\| = O(\chi \sqrt{Tn})$.


\section{The runtime of the Na\"iveBt backtracking algorithm for $k$-SAT}
\label{app:ordereddll}

Here we find simple, yet fairly precise, bounds on the expected number of vertices in the backtracking tree for the Na\"iveBt algorithm when applied to random $k$-SAT with $k=O(1)$.

Assume we pick a random instance of $k$-SAT by choosing $m=\alpha n$ clauses, for some $\alpha$ such that $1 \le \alpha \le n^{k-1}$. Each clause contains $k$ distinct variables, and each variable can be present negated or unnegated. Each clause is chosen uniformly at random, with replacement, from the set of all $2^k \binom{n}{k}$ allowed such clauses. Recall that the Na\"iveBt algorithm is a backtracking algorithm in the framework of Algorithm \ref{alg:backtrack} where the heuristic $h$ simply picks the lowest index of a variable which has not yet been assigned a value.

The probability that a given assignment to variables $x_1,\dots,x_\ell$ is consistent with all the clauses in such a random instance is
\beas && \left(1 - \Pr[\text{clause only depends on first $\ell$ variables}] \Pr[\text{assignment fails to satisfy clause}] \right)^m \\
&=& \left(1 - \frac{\binom{\ell}{k}}{2^k \binom{n}{k}}\right)^m.
\eeas
The expected number of solutions is
\[ 2^n (1- 2^{-k})^m \le 2^n e^{-2^{-k}\alpha n} = 2^{(1-(2^{-k} \ln 2)\alpha)n}, \]
so for $\alpha \gg 2^k/\ln 2$, the probability that the instance is satisfiable is exponentially small in $n$. The expected number of vertices in the backtracking tree is
\[ E := \sum_{\ell=0}^n 2^\ell \left(1 - \frac{\binom{\ell}{k}}{2^k \binom{n}{k}} \right)^m. \]
We have
\[ \max_{\ell} 2^\ell  \left(1 - \frac{\binom{\ell}{k}}{2^k \binom{n}{k}} \right)^m \le E \le (n+1) \max_{\ell} 2^\ell  \left(1 - \frac{\binom{\ell}{k}}{2^k \binom{n}{k}}\right)^m. \]
For the upper bound, using $1-x \le e^{-x}$ we obtain
\be \label{eq:max} \max_{\ell} 2^\ell \left(1 - \frac{\binom{\ell}{k}}{2^k \binom{n}{k}} \right)^m \le \max_{\ell} 2^\ell e^{-\alpha 2^{-k} n \frac{\ell(\ell-1)\dots(\ell-k+1)}{n(n-1)\dots(n-k+1)} } \le \max_{\ell} 2^\ell e^{-\alpha 2^{-k} \ell^k n^{1-k} (1-k^2/\ell) }. \ee
By taking the derivative over $\ell$ we get that the maximum is achieved for $\ell$ such that
\[ \ell\left(1 - \frac{k(k-1)}{\ell} \right)^{1/(k-1)} = n \left(\frac{2^k \ln 2}{\alpha k}\right)^{1/(k-1)}. \]
The left-hand side is of the form $\ell(1 - O(1/\ell)) = \ell - O(1)$, so we can ignore the $O(1/\ell)$ correction term as it can only change the final bound by an $O(1)$ factor. Rounding $\ell$ to the nearest integer similarly does not change the asymptotic complexity. So, inserting the right-hand side in (\ref{eq:max}) and simplifying, we obtain
%
%
%
%
\[ E = O(n 2^{Cn}), \text{ where } C = \left( \frac{2^k \ln 2}{\alpha k} \right)^{1/(k-1)} \left(1 - \frac{1}{k} \right). \]
On the other hand, inserting this value of $\ell$ into the lower bound and using the inequality $1-x \ge e^{-x-x^2}$ gives
\[ E = \Omega(2^\ell e^{-\alpha 2^{-k} \ell^k n^{1-k} - \alpha 2^{-2k} \ell^{2k} n^{1-2k}} ) = \Omega(2^{C'n}), \]
where
\[ C' = \left( \frac{2^k \ln 2}{\alpha k} \right)^{1/(k-1)} \left(1 - \frac{1}{k} - \frac{\ln 2}{\alpha k^2} \left( \frac{2^k \ln 2}{\alpha k} \right)^{1/(k-1)} \right). \]
For example, for $k = 3$, we have $C \le 0.907/\sqrt{\alpha}$, $C' \ge 0.906/\sqrt{\alpha} - 0.142/\alpha^2$. 


\bibliographystyle{plain}
\bibliography{../../thesis}

\end{document}